\documentclass[aps,prd,reprint,nofootinbib,showpacs]{revtex4-1}
\usepackage{amssymb}
\usepackage{amsmath}
\usepackage{amsfonts}
\usepackage{amsthm}
\usepackage{graphicx}
\usepackage{xcolor}
\newtheorem{proposition}{Proposition}[section]

\begin{document}
\title{Effects of non-linear vacuum electrodynamics on the polarization plane of light}
\author{Volker Perlick}
\email{volker.perlick@zarm.uni-bremen.de}
\affiliation{ZARM, University of Bremen, Am Fallturm, 28359 Bremen, Germany}
\author{Claus L\"{a}mmerzahl}
\email{claus.laemmerzahl@zarm.uni-bremen.de}
\affiliation{ZARM, University of Bremen, Am Fallturm, 28359 Bremen, Germany \\
Institute of Physics, University of Oldenburg, 26111 Oldenburg, 
Germany}{}
\author{Alfredo Mac{\' \i}as}
\email{amac@xanum.uam.mx}
\affiliation{Departamento de F{\' \i}sica, Universidad Aut{\' o}noma Metropolitana-Iztapalapa, 
A.P. 55-534, Mexico D.F. 09340, Mexico}{}
\pacs{03.50.Kk,11.10.Lm}

\begin{abstract}
We consider the Pleba{\'n}ski class of nonlinear theories of vacuum 
electrodynamics, i.e., Lagrangian theories that are Lorentz invariant
and gauge invariant. Our main goal is to derive the transport law of 
the polarization plane in such a theory, on an unspecified general-relativistic
spacetime and with an unspecified electromagnetic background field. To that
end we start out from an approximate-plane-harmonic-wave ansatz that takes the 
generation of higher harmonics into account. By this ansatz, the 
electromagnetic field is written as an asymptotic series with respect to
a parameter $\alpha$, where the limit $\alpha \to 0$ corresponds to 
sending the frequency to infinity. We demonstrate that by solving the
generalized Maxwell equations to zeroth and first order with respect to
$\alpha$ one gets a unique transport law for the polarization plane along
each light ray. We exemplify the general results with the Born-Infeld 
theory.

\end{abstract}

\maketitle

\section{Introduction}\label{sec:intro}

In the standard Maxwell theory of vacuum electrodynamics, the
field strength tensor $F_{ab}$ (which comprises the vector fields
$\vec{E}$ and $\vec{B}$ ) is related to the excitation tensor
(which comprises the vector fields $\vec{D}$ and $\vec{H}$) by a 
linear constitutive law. However, there are good reasons to
assume that this law has to be replaced by a nonlinear relation
for very strong electromagnetic fields. In the course of history,
several such nonlinear modifications of the vacuum Maxwell theory
have been suggested. 

One of the best known examples is the theory of Born and Infeld \cite{BornInfeld1934} 
from 1934. Its introduction was motivated by the observation that in the 
standard Maxwell vacuum theory the field energy in an arbitrarily small ball 
around a point charge is infinite which leads to an infinite self-force, and 
that this infinity might be overcome if one modifies the constitutive law of 
the vacuum in a nonlinear fashion. The Born-Infeld theory introduces a new
hypothetical constant of Nature, $b_0$, with the dimension of a (magnetic) field 
strength. In the limit $b_0 \to \infty$ the theory approaches the standard Maxwell 
theory, i.e., the fact that the latter is in good agreement with experiments can 
be understood if one assumes that $b_0$ is very large. On the basis of the 
Born-Infeld theory one would have to expect measurable deviations from the
vacuum Maxwell theory in electromagnetic fields that are of a similar order
of magnitude as $b_0$. Although not exactly in the main stream of physics, the
Born-Infeld theory was always taken seriously by many scientists. In the late 1990s
this theory got an additional strong push when Tseytlin \cite{Tseytlin1999}
realized that it can be derived, as an effective theory, from some kind of 
string theories.

Another very well known nonlinear modification of the vacuum Maxwell theory is the 
Heisenberg-Euler theory \cite{HeisenbergEuler1936} from 1936. It is a 
classical field theory which comes about, as an effective theory, if one-loop
corrections from quantum electrodynamics are taken into account. In contrast to the
Born-Infeld theory, it does not involve any new hypothetical constant of Nature, i.e., 
it numerically predicts how strong an electromagnetic field has to be in order
to produce measurable deviations from the standard vacuum Maxwell theory. Since
a few years (magnetic) fields of this strength can be produced in the 
laboratory. 

The Born-Infeld theory and the Heisenberg-Euler theory are Lorentz invariant,
they are gauge invariant, and they derive from a Lagrangian. The entire class of 
theories that share these properties was systematically studied by Pleba{\'n}ski
\cite{Plebanski1970}, with important early contributions by Boillat \cite{Boillat1970}.
We refer to it as to the \emph{Pleba{\'n}ski class} of electromagnetic theories. 
The Born-Infeld theory and the Heisenberg-Euler theory are the best known examples 
in this class, but there are many more. In particular, there are theories of the 
Pleba{\'n}ski class that allow for regular black-hole solutions if they are coupled 
to Einstein's field equations. The first two examples were found by Ay{\'o}n-Beato 
and Garc{\'\i}a \cite{AyonGarcia1998,AyonGarcia1999}.  

It is a general feature of nonlinear theories that the superposition principle is 
no longer satisfied. As a consequence, the propagation of light is influenced
by electromagnetic background fields. This effect is known as ``light-by-light
scattering'' and it has been observed in 1997, see Burke at el. \cite{BurkeEtAl1997}, 
in good agreement with the prediction by the Heisenberg-Euler theory. Another effect 
predicted by most theories of the Pleba{\'n}ski class, with the notable exception
of the Born-Infeld theory, is birefringence in vacuo. This means that, according to 
these theories, a light beam that enters into a region with a sufficiently
strong electromagnetic background field would split into two beams, as in the 
case of a light beam entering into a crystal according to ordinary optics, with
the two different beams corresponding to two different polarization states. Such a 
birefringence in vacuo is predicted, in particular, by the Heisenberg-Euler theory.  
Experimentalists are trying to observe this effect since several years and there is 
the general expectation that these attempts will be succesful soon, see in particular
the most recent status report on the so-called PVLAS experiment by 
Della Valle et al. \cite{DellavalleEtAl2016}. In this experiment not only the 
birefringence in vacuo but also the dichroism of the Heisenberg-Euler theory 
is tried to be measured. The latter means the effect that there are different 
absorption coefficients for the two different polarization states which results 
in an apparent rotation of the polarization plane. Finally, we mention that there 
are also attempts to verify effects from nonlinear electrodynamics with astrophysical 
observations. A particularly promising idea is to observe the birefringence in vacuo 
if light passes through a very strong magnetic field, such as in the neighborhood of 
a magnetar. A first observation that might indicate such an effect was already made, 
see Mignani et al. \cite{MignaniEtAl2017}.

We emphasize that experiments searching for birefringence in vacuo cannot be used as 
tests for the Born-Infeld theory because in the latter there is no such effect. As 
an alternative, the Born-Infeld theory may be tested with the help of Michelson 
interferometry. Such an experiment was discussed for the Heisenberg-Euler theory by 
Boer and van Holten \cite{BoerHolten2002}, D{\"o}brich and Gies \cite{DoebrichGies2009},
Zavattini  and  Calloni \cite{ZavattiniCalloni2009} and Grote \cite{Grote2015},
for the Heisenberg-Euler and the Born-Infeld theories by Denisov, Krivchenkov and  
Kravtsov \cite{DenisovKrivchenkovKravtsov2004}, and in detail for a general theory of 
the Pleba{\'n}ski class by Schellstede et al. \cite{SchellstedePerlickLaemmerzahl2015}. 
Moreover, there are suggestions to test the Born-Infeld theory with wave-guides, 
see Ferraro \cite{Ferraro2007}, or with fluid motions in a magnetic 
background field, see Dereli and Tucker \cite{DereliTucker2010}. As of now, none 
of these experiments has been actually carried through.  

In this paper we want to study, for a general theory of the Pleba{\'n}ski class,
the effect of a background field on the transport law of the polarization plane
along a light ray. This will give us a new way of testing these theories, in
particular the Born-Infeld theory, experimentally. We emphasize that this is
to be distinguished from all the experiments mentioned above. In particular,
it is not to be confused with the planned observation of dichroism by the 
PVLAS experiment: The latter is an effect on the absorption of light, 
depending on the polarization state. Here we want to investigate the direct
effect of an electromagnetic background field on the polarization plane.

To that end we start out from an approximate-plane-harmonic-wave ansatz, 
taking the generation of higher harmonics into account. By this ansatz the 
electromagnetic field is written as an asymptotic series with respect to a 
parameter $\alpha$. Sending $\alpha$ to zero corresponds to sending the frequency 
to infinity. We will see that we have to consider the generalized Maxwell 
equations to zeroth order and to first order with respect to $\alpha$ in order 
to determine the transport law for the polarization plane. Earlier studies of 
the high-frequency limit in nonlinear theories were restricted to the derivation 
of the eikonal equation from the zeroth order of the generalized Maxwell equations. 
It is well known that, as a result, one finds that the light rays are the null 
geodesics of two optical metrics; this was first shown by Novello et 
al. \cite{Novelloetal2000} and later, in different representations, by Obukhov 
and Rubilar \cite{ObukhovRubilar2002} and by Schellstede et al. 
\cite{SchellstedePerlickLaemmerzahl2015}. 
To the best of our knowledge, the transport law of the polarization plane was not 
yet considered for an arbitrary theory of the Pleba{\'n}ski class. 

We will not specify the nonlinear electromagnetic theory, apart from the 
fact that we require it to be of the Pleba{\'n}ski class. However, we mention 
that not all theories of this type are to be considered as physically meaningful: 
Some of them violate causality in the sense that the light cones of the optical 
metrics are not inside the light cone of the spacetime metric, see Schellstede
et al. \cite{SchellstedePerlickLaemmerzahl2016}. Also, not all of them
give rise to a well-posed initial-value problem, see Abalos et al. \cite{Abalosetal2015}.

The paper is organized as follows. In Section \ref{sec:Pleb} we briefly review
the basic features of theories of the Pleba{\'n}ski class. In Section 
\ref{sec:phw} we introduce our approximate-plane-wave ansatz on an arbitraty 
general-relativistic spacetime and for an arbitrary electromagnetic
background field. In Section \ref{sec:zerothorder} we evaluate the generalized
Maxwell equations to zeroth order which gives us the eikonal equation and an 
algebraic condition on the polarization plane. In Section \ref{sec:firstorder}
we consider the generalized Maxwell equations to first order and discuss the
additional conditions they give us on the polarization plane. In Section 
\ref{sec:BI} we exemplify the results with the Born-Infeld theory.

\section{The Pleba{\'n}ski class of non-linear electrodynamical theories}\label{sec:Pleb}
We consider a general-relativistic spacetime, i.e., an oriented 4-dimensional manifold
with a metric tensor $g_{ab}$ of Lorentzian signature. The covariant derivative 
associated with the Levi-Civita connection of the metric will be denoted $\nabla _a$. 
Latin indices take values 0,1,2,3 and are lowered with $g_{ab}$ and 
raised with its inverse $g^{bc}$.

The Pleba{\'n}ski class \cite{Plebanski1970}
consists of all non-linear electrodynamical theories that derive from an
action of the form
\begin{equation}\label{eq:action}
S[A_c]=\frac{1}{4\pi c}\int _M \left(\mathcal{L}(F,G)+
\frac{4\pi}{c}\,j^aA_a \right)\, \sqrt{| \mathrm{det}(g_{bc})|} \; d^4x \,.
\end{equation}
Here $M$ is a domain of the spacetime, $d^4x=dx^0 \wedge dx^1 \wedge dx^2 \wedge dx^3$,
 $j^a$ is a \emph{given} current density, $A_a$ is the electromagnetic potential, 
\begin{equation}\label{eq:Fab}
F_{ab} = \nabla _a A_b-\nabla _b A_a
\end{equation}
is the electromagnetic field strength and 
$\mathcal{L}$ is the Lagrangian for the electromagnetic field. It is assumed that the
latter depends only on the two invariants 
\begin{equation}\label{eq:FG}
F=\frac{1}{2}\,F_{ab}F^{ab} 
\quad \text{and} \quad
G=-\frac{1}{4}\,F_{ab} \, ^{\star \!} F^{ab} \, .
\end{equation}
Here and in the following, $^{\star \,}$denotes the
Hodge star operator, i.e., 
\begin{equation}
^{\star \!} F_{ab}\, = \, \frac{1}{2} \, \varepsilon_{abcd}F^{cd}
\end{equation}
where $\varepsilon _{abcd}$ is the totally antisymmetric Levi-Civita
tensor field (volume form) associated with the spacetime metric.   

By (\ref{eq:Fab}) the homogeneous Maxwell equation is automatically satisfied,
\begin{equation}\label{eq:Max1}
\varepsilon ^{abcd} \nabla_b F_{cd}=0 \, .
\end{equation}
Requiring that the variational derivative of the action (\ref{eq:action}) with respect 
to the potential $A_c$ vanishes, for all compact domains $M$ and all variations that 
keep $A_c$ fixed on the boundary of $M$, leads to the inhomogeneous Maxwell equation,
\begin{equation}\label{eq:Max2}
\nabla_a H^{ab} \, = \, - \, \frac{4\pi}{c}\,j^b\,,
\end{equation}
where
\begin{equation}\label{eq:const}
H^{ab}=-\frac{\partial \mathcal{L}}{\partial F_{ab}}
= 
 -2\, \mathcal{L}_F \, F^{ab}+\mathcal{L}_G \, ^{\star \!} F^{ab} 
\end{equation}
is the electromagnetic excitation. For the sake of brevity, we write
\begin{equation}\label{eq:LFLG}
\mathcal{L}_F = \dfrac{\partial \mathcal{L}}{\partial F} \, , \quad
\mathcal{L}_G = \dfrac{\partial \mathcal{L}}{\partial G} 
\end{equation}
and 
\begin{equation}\label{eq:LFLG2}
\mathcal{L}_{FF} = \dfrac{\partial ^2 \mathcal{L}}{\partial F ^2} \, , \quad
\mathcal{L}_{GG} = \dfrac{\partial ^2 \mathcal{L}}{\partial G ^2} \, , \quad
\mathcal{L}_{FG} = \dfrac{\partial ^2 \mathcal{L}}{\partial F \partial G} \, .
\end{equation}
It is the constitutive law (\ref{eq:const}) that 
distinguishes  different theories, while the Maxwell equations 
(\ref{eq:Max1}) and (\ref{eq:Max2}) are always the same.

Each particular theory of the Pleba{\'n}ski class is characterized by a 
particular Lagrangian and, thereby, by a particular constitutive law. Let us
mention the two most important examples: For the Born-Infeld theory \cite{BornInfeld1934}, 
the Lagrangian reads
\begin{equation}\label{eq:LBI}
\mathcal{L} = b_0^2 - b_0^2 \sqrt{1+ \dfrac{F}{b_0^2}-\dfrac{G^2}{b_0^4}}
\end{equation}
where $b_0$ is a hypothetical constant of Nature with the dimension of
a magnetic field strength. For $b_0 \to \infty$ the Born-Infeld theory
reproduces the standard Maxwell vacuum theory. 
For the Heisenberg-Euler theory \cite{HeisenbergEuler1936},
\begin{equation}\label{eq:LHE}
\mathcal{L} = E_0^2 \left( - \dfrac{1}{2} \dfrac{F}{E_0^2} + 
\Lambda \Big( \dfrac{F^2}{E_0^4} + 7 \dfrac{G^2}{E_0^4} \Big) + \, \dots \right) 
\end{equation}
where $E_0=m^2c^4/e^3$ and $\Lambda = \hbar c/(90 \pi e^2)$. Here $m$ is
the electron mass, $e$ is the electron charge, $\hbar$ is the reduced Planck
constant and the ellipses in (\ref{eq:LHE})   stand for terms of third and 
higher order in $F$ and $G$.
 
\section{Approximate-plane-harmonic-wave~ansatz}\label{sec:phw}
An approximate-plane-harmonic wave is a one-parameter family $F^{\alpha}_{cd}$
of field strength tensors, depending on a real parameter $\alpha$, of the form
\begin{equation}\label{eq:aphwF}
F^{\alpha}_{cd} = F_{cd} + \alpha \, F_{cd}^{(1)} + 
\sum _{K=2}^{\infty} \alpha ^K F_{cd} ^{(K)}
\end{equation}
where 
\begin{equation}\label{eq:f1}
F_{cd}^{(1)} = \mathrm{Re} \Big\{ e^{iS/\alpha} f_{cd} ^{(11)} \Big\}
\end{equation}
and
\begin{equation}\label{eq:f}
F_{cd}^{(K)} = \sum _{\tilde{K} =0} ^{K} 
\mathrm{Re} \Big\{ e^{i \tilde{K} S/\alpha} f_{cd} ^{(K \tilde{K})} \Big\}
\, \quad \text{for} \; K \ge 2 \, .
\end{equation}
Here $F_{cd}$ is a given electromagnetic background field that is independent of $\alpha$, 
$S$ is a real-valued function and $f^{(K \tilde{K})}_{cd}$ is a complex-valued
antisymmetric tensor field  for each pair of integers ${K, \tilde{K}}$ that occurs. 
We assume that, on the spacetime region considered, the tensor 
fields $\nabla _a S$ and $f_{cd}^{(11)}$ have no zeros. The series is to be 
understood as an asymptotic series, not as a convergent series. 

The function $S$ is called the \emph{eikonal function}. On a sufficiently
small neighborhood, the field $F^{(1)}_{ab}$ is approximately a plane 
harmonic wave: The surfaces $S= \mathrm{constant}$ are the wave-fronts 
and the gradient of S divided by alpha defines the wave four-covector.
Correspondingly, the frequency measured by an observer with four-velocity $U^a$ is 
$\omega = U^a \nabla _aS / \alpha$. The limit $\alpha \to 0$ corresponds 
to sending the frequency to infinity. The idea is to feed the ansatz (\ref{eq:aphwF})
into Maxwell's equations, to solve these equations iteratively order by
order in $\alpha$ and, in this way, to asymptotically approach a one-parameter
family of exact solutions. 

Our ansatz (\ref{eq:aphwF}) is a generalization of the standard 
approximate-plane-harmonic-wave ansatz. The latter goes back to Ralph Luneburg
and is detailed, for wave propagation in linear and isotropic media, e.g. in the
text-book by Kline and Kay \cite{KlineKay1965}. Our ansatz is more general
in two respects: Firstly, we take a non-zero background field into account.
In a linear theory, it suffices to consider the case with zero background
field because, by the superposition principle, the propagation of the 
approximate-plane-harmonic wave is independent of a background field. In a 
non-linear theory, however, the propagation is influenced by a background field.
Secondly, the higher-order fields, $F_{cd}^{(K)}$ for $K \ge 2$, come not 
only with the same frequency as the first-order field $F_{cd}^{(1)}$ (the
terms with $\tilde{K} = 1$) but also with integer multiples of this frequency (the 
terms with $\tilde{K} \neq 1$). This reflects the \emph{generation of higher harmonics} 
which is well-known from optics in non-linear media. It should not come as a 
surprise that it has to be taken into account also in the non-linear vacuum 
theories of the Pleba{\'n}ski class. Higher harmonics play no role if one 
considers Maxwell's equations only to the lowest order (i.e., $\alpha ^0$). 
This is the reason why it was not necessary to take them into account 
in \cite{Perlick2011} where the eikonal equation was derived for Maxwell's 
equations with a local but otherwise arbitrary constitutive law. In the 
present paper, however, we want to derive the transport law for the polarization 
plane which requires considering Maxwell's equations also to the next order 
(i.e. $\alpha ^1$). We will see that these equations cannot in general be 
solved if we set all terms $f_{cd}^{(K \tilde{K})}$ with $\tilde{K} \neq 1$ equal to zero.

For our purpose we need the series (\ref{eq:aphwF}) up to second order,
\begin{gather}
F^{\alpha}_{cd} = F_{cd} + 
\alpha \mathrm{Re} \Big\{ e^{iS/ \alpha} f_{cd}^{(11)} \Big\} \qquad
\nonumber
\\
+ \alpha ^2 \mathrm{Re} \Big\{ f_{cd}^{(20)} + 
e^{iS/ \alpha} f_{cd}^{(21)} +
e^{2iS/ \alpha} f_{cd}^{(22)} \Big\} + \, \dots
\label{eq:aphwF2}
\end{gather}
which includes \emph{frequency doubling} $(\tilde{K}=2)$ and the generation of a 
non-oscillatory mode, known from non-linear media as \emph{optical 
rectification} $(\tilde{K}=0)$. 
The homogeneous Maxwell equation (\ref{eq:Max1}) is automatically 
satisfied for all $\alpha$ if we assume that (\ref{eq:aphwF2}) 
derives from a potential,
\begin{equation}\label{eq:Aalpha}
F^{\alpha}_{cd} = \nabla _c A^{\alpha}_d- \nabla _d A^{\alpha}_c  \, .
\end{equation}
It is easy to see that such a potential (up to an arbitrary gradient
term) must be of the form
\begin{gather}
A^{\alpha}_{d} = A_{d} + 
\alpha ^2 \mathrm{Re} \Big\{ a_d^{(10)} + e^{iS/ \alpha} a_d^{(11)} \Big\} 
\qquad \quad 
\nonumber
\\
+ \alpha ^3 \mathrm{Re} \Big\{ a_d^{(20)} + 
e^{iS/ \alpha} a_d^{(21)} +
e^{2iS/ \alpha} a_d^{(22)} \Big\} + \, \dots
\label{eq:aphwA0}
\end{gather}
Then (\ref{eq:Aalpha}) holds to zeroth order in $\alpha$ with
\begin{equation}\label{eq:Fa}
F_{cd} = \nabla _c A_d- \nabla _d A_c \, , 
\end{equation}
to first order with
\begin{equation}\label{eq:f11a}
f_{cd}^{(11)} = i \big( \nabla _c S \, a^{(11)}_d-\nabla _d S \, a^{(11)}_c \big) \, ,
\end{equation}
and to second order with
\begin{gather}
f_{cd}^{(20)}  = \nabla _c a^{(10)}_d - \nabla _d a^{(10)}_d \, ,
\label{eq:f20a}
\\
f_{cd}^{(21)}  = \nabla _c a^{(11)}_d - \nabla _d a^{(11)}_d 
\nonumber
\\
+ i \big( \nabla_c S \, a_d^{(21)} - \nabla_d S \, a_c^{(21)} \big) \, ,
\label{eq:f21a}
\\
f_{cd}^{(22)}  = 
2 \, i \big( \nabla_c S \, a_d^{(22)} - \nabla_d S \, a_c^{(22)} \big) \, .
\label{eq:f22a}
\end{gather} 
Here we have used our assumption that the gradient of $S$ has no zeros which
implies that $S \neq 0$ almost everywhere and that, accordingly, the functions 
$1$, $\mathrm{sin} \big( S(x)/\alpha \big)$, $\mathrm{cos}(\big( S(x)/\alpha \big))$, 
$\mathrm{sin} \big( 2 S(x)/\alpha \big))$ and $\mathrm{cos} \big(2 S(x)/\alpha \big)$ 
are linearly independent.

Feeding the approximate-plane-harmonic  wave (\ref{eq:aphwF}) into
the constitutive law (\ref{eq:const}) gives, after a rather long but 
straight-forward calculation, an excitation of the form
\begin{gather}\label{eq:aphwH}
H^{\alpha}_{ab} = H_{ab} +  
\alpha \, \mathrm{Re} \Big\{ e^{iS/\alpha} h^{(11)}_{ab} \Big\} 
\qquad 
\\
\nonumber
+
\alpha ^2 \, \mathrm{Re} \Big\{ h^{(20)}_{ab} 
+ e^{iS/\alpha} h^{(21)}_{ab} +
e^{i2S/\alpha} h^{(22)}_{ab}
\Big\} + \, \dots
\end{gather}
The zeroth order term in (\ref{eq:aphwH}) is just the excitation of the background
field, 
\begin{equation}\label{eq:aphwH0}
H_{ab} =
 -2\, \mathcal{L}_F \, F_{ab}+\mathcal{L}_G {} ^{\star \!} F_{ab} \, ,
\end{equation}
the first-order amplitude is 
\begin{equation}\label{eq:h11}
h_{ab}^{(11)} = \dfrac{1}{2} \chi _{ab}{}^{cd} f_{cd}^{(11)} \, ,
\end{equation}
and the second-order amplitudes are
\begin{equation}\label{eq:h20}
h_{ab}^{(20)} =  
\dfrac{1}{2} \chi _{ab}{}^{cd} f_{cd}^{(20)}
+ \dfrac{1}{2} \psi _{ab}{}^{cdef} f_{cd}^{(11)} \overline{f}{}_{ef}^{(11)}\, ,
\end{equation}
\begin{equation}\label{eq:h21}
h_{ab}^{(21)} = \dfrac{1}{2} \chi _{ab}{}^{cd} f_{ef}^{(21)} \, ,
\end{equation}
\begin{equation}\label{eq:h22}
h_{ab}^{(22)} = \dfrac{1}{2} \chi _{ab}{}^{cd} f_{cd}^{(22)} +
\dfrac{1}{2} \psi _{ab}{}^{cdef} f_{cd}^{(11)} f_{ef}^{(11)}\, ,
\end{equation}
with
\begin{gather}
\chi _{ab}{}^{cd} = \mathcal{L} _G \varepsilon _{ab}{}^{cd}
- 2 \mathcal{L} _F 
\big( \delta _a^c \delta _b^d-\delta _a^d \delta _b^c \big)
-4 \mathcal{L} _{FF} F_{ab}F^{cd} 
\nonumber
\\
\label{eq:chi}
+ 2 \mathcal{L}_{FG} \big( F_{ab}{} ^{\star \!}F^{cd} +{}
^{\star \!}F_{ab} F^{cd} \big) 
- \mathcal{L}_{GG}{} ^{\star \!}F_{ab}{} ^{\star \!}F^{cd} 
\end{gather}
and
\begin{gather}
\psi _{ab}{}^{cdef} = 
\dfrac{1}{4} \Big(-2 \mathcal{L}_{FF}F_{ab} 
+\mathcal{L}_{FG} {} ^{\star \!}F_{ab} \Big)
\Big( g^{ce}g^{df}-g^{de}g^{cf} \Big)
\nonumber
\\
+\dfrac{1}{2} \Big( \delta _a^c \delta _b^d -\delta _b^c \delta _a^d \Big)
\Big(-2 \mathcal{L}_{FF}F^{ef}+ 
\mathcal{L}_{FG} {} ^{\star \!}F^{ef} \Big)
\nonumber
\\
- \dfrac{1}{8} \Big(-2 \mathcal{L}_{FG}F_{ab} 
+ \mathcal{L}_{GG} {} ^{\star \!}F_{ab} \Big) \varepsilon ^{cdef} 
\nonumber
\\
- \dfrac{1}{4} \varepsilon _{ab}{}^{cd} \Big(-2 \mathcal{L}_{FG}F^{ef} 
+\mathcal{L}_{GG} {} ^{\star \!}F^{ef} \Big)  
\nonumber
\\
+ \dfrac{1}{2} \Big(-2 \mathcal{L}_{FFF}F_{ab} + 
\mathcal{L}_{FFG} {} ^{\star \!}F_{ab} \Big) F^{cd}F^{ef}
\nonumber
\\
- \dfrac{1}{2} \Big(-2 \mathcal{L}_{FFG}F_{ab} + 
\mathcal{L}_{FGG} {} ^{\star \!}F_{ab} \Big) {} F^{cd} {} ^{\star \!}F^{ef} 
\nonumber
\\
+ \dfrac{1}{8} \Big(-2 \mathcal{L}_{FGG}F_{ab} + 
\mathcal{L}_{GGG} {} ^{\star \!}F_{ab} \Big) {} ^{\star \!}F^{cd} {} ^{\star \!}F^{ef} 
\, .
\end{gather}

We see that the first- order constitutive law (\ref{eq:h11}) is of the 
same form as the constitutive law of a linear medium, but now with a constitutive
tensor $\chi _{ab}{}^{cd}$ that depends on the invariants $F$ and $G$
of the background field.  Quite generally, such a constitutive tensor
can be decomposed into principal part, skewon part and axion part
(see Hehl and Obukhov \cite{HehlObukhov2003}). In (\ref{eq:chi}),
the first term is the axion part, the rest is the principal part and the 
skewon part is zero. It is known  \cite{HehlObukhov2003} that the
skewon part is always vanishing if the theory derives from a variational
principle.

At the second order, we get for each of the three amplitudes 
$h_{ab}^{(2\tilde{K})}$ a linear law with the same constitutive tensor 
$\chi_{ab}{}^{cd}$ as for the first order, but for $\tilde{K}=0$ and $\tilde{K}=2$ 
additional quadratic terms with a second-order constitutive tensor
$\psi _{ab}{}^{cdef}$ which looks rather complicated.

We will now evaluate the Maxwell equations. The homogeneous Maxwell
equation is satisfied if we express the amplitudes $f_{cd}^{(K \tilde{K})}$
in terms of the potential according to (\ref{eq:f11a}), (\ref{eq:f20a}),
(\ref{eq:f21a}) and (\ref{eq:f22a}). Feeding the excitation (\ref{eq:aphwH})
into the inhomogeneous Maxwell equation requires at zeroth order
\begin{equation}\label{eq:Maxzero1}
- \dfrac{4 \pi}{c} j_b = \nabla ^a H_{ab}  \, ,
\end{equation}
\begin{equation}\label{eq:Maxzero2}
0 = \nabla ^a S \, h_{ab} ^{(11)} \, ,
\end{equation}
and at first order
\begin{equation}\label{eq:Maxfirst1}
0 = \nabla ^a h_{ab}^{(11)}  + i \nabla ^aS \, h_{ab}^{(21)} \, ,
\end{equation}
\begin{equation}\label{eq:Maxfirst2}
0 = \nabla ^a S \, h_{ab}^{(22)}  \, .
\end{equation}
Here we have assumed that the current $j_b$ is independent of $\alpha$,
i.e., that only the background field may have a source whereas our 
approximate-plane-harmonic wave is source-free. Moreover, 
we have again used our assumption that the gradient of $S$ has no zeros which
implies that the functions 
$1$, $\mathrm{sin} \big( S(x)/\alpha \big)$, $\mathrm{cos}(\big( S(x)/\alpha \big))$, 
$\mathrm{sin} \big( 2 S(x)/\alpha \big))$ and $\mathrm{cos} \big(2 S(x)/\alpha \big)$ 
are linearly independent.

At zeroth order we get one equation, (\ref{eq:Maxzero2}), that has to be 
satisfied. With (\ref{eq:f11a}) and (\ref{eq:h11}) this equation reads
\begin{equation}\label{eq:Maxzeroa}
0 = 
\nabla ^a S \chi _{ab}{}^{cd} \nabla _c S \, a_d^{(11)} \, .
\end{equation}
We will evaluate this equation in the next section. We will see that
it gives us the \emph{eikonal equation} for $S$ and an algebraic condition
on $a_d^{(11)}$ which is known as the zeroth order \emph{polarization condition}. 
Note that $a_d^{(11)}$ is not gauge-invariant: As can be read from (\ref{eq:f11a}),
the field strength $f_{cd}^{(11)}$ is unchanged if a multiple  of $\nabla _d S$ is 
added to $a^{(11)}_d$. We will see that the zeroth order polarization condition
is actually a condition on the (gauge-invariant) plane spanned by $a_d^{(11)}$
and $\nabla _d S$. We refer to this plane as to the \emph{polarization plane}.

At first order we get two equations, (\ref{eq:Maxfirst1}) and (\ref{eq:Maxfirst2}),
that have to be satisfied. With (\ref{eq:f11a}), (\ref{eq:f20a}), (\ref{eq:f21a}), 
(\ref{eq:h11}), (\ref{eq:h21}) and (\ref{eq:h22}) these equations 
read
\begin{gather}
0 = \nabla ^a \Big( \chi _{ab}{}^{cd} \nabla _c S \, a_d^{(11)} \Big)
\nonumber
\\
+  \nabla ^a S \, \chi _{ab}{}^{cd} \nabla _c a_d^{(11)}
+ i \nabla ^a S \, \chi _{ab}{}^{cd} \nabla _c S \, a_d ^{(21)} \, ,
\label{eq:Max2first1a}
\end{gather}
\begin{equation}\label{eq:Max2first2a}
0 = \nabla ^a S \, \chi _{ab}{}^{cd} \nabla _c S \, a_d ^{(22)}
- \psi _{ab}{}^{cdef} \nabla ^a S \, \nabla _c S \, a_d ^{(11)} a_f^{(11)} \, . 
\end{equation}
We will evaluate these two equations, as far as necessary for our
purpose, in Section \ref{sec:firstorder} below. They will give us a differential 
equation for $a_d^{(11)}$ which is known as the first-order \emph{transport equation}
and algebraic conditions on $a_d^{(21)}$ and $a_d^{(22)}$ which are the first-order  
polarization conditions. These equations are not in general satisfied if $a_d^{(22)}=0$,
i.e., frequency doubling has to be taken into account if Maxwell's equations are
to be solved to first order.

If one wants to go beyond the first order, one can do this step by step. At the
$K^{\mathrm{th}}$ level one gets transport equations for the amplitudes 
$a_d^{(K \tilde{K})}$ and polarization conditions on the amplitudes 
$a_d^{((K+1)\tilde{K})}$.

\section{Evaluation of the zeroth-order field equation}\label{sec:zerothorder}

In Section \ref{subsec:eikonal} we will derive the eikonal equation  
from the zeroth-order field equation (\ref{eq:Maxzeroa}), in
Section \ref{subsec:rays} we will determine the Hamiltonian for the rays
and in Section \ref{subsec:polcon} we will evaluate the zeroth-order
polarization condition. The main results of Sections \ref{subsec:eikonal}
and \ref{subsec:rays} are not new. In particular, it is known that for any theory of 
the Pleba{\'n}ski class the rays are the null geodesics of two optical 
metrics. This was first demonstrated by Novello et al. \cite{Novelloetal2000}. 
The same result was re\-de\-rived, using a different representation, by Obukhov 
and Rubilar \cite{ObukhovRubilar2002} who also showed that the optical
metrics have Lorentzian signature if they are non-degenerate. Still
another form of the optical metrics was derived by Schellstede et al. 
\cite{SchellstedePerlickLaemmerzahl2015}. However, we have to rederive these 
known results here because in doing so we will also establish a number of 
new relations that will be needed later. We will use the same representation 
as in \cite{SchellstedePerlickLaemmerzahl2015}.

\subsection{Derivation of the eikonal equation}\label{subsec:eikonal}

In the following we write
\begin{equation}\label{eq:puv}
p_a= \nabla _a S \, , \quad
u_a= F_{ab} \nabla ^b S \, , \quad
v_a= {\,}{}^{*}{\!}F{}_{ab} \nabla ^b S  
\end{equation}
which implies 
\begin{equation}\label{eq:puvorth}
p_au^a = p_a v^a = 0 \, .
\end{equation}
Then the zeroth-order field equation (\ref{eq:Maxzeroa}) can be 
rewritten as
\begin{equation}\label{eq:MA}
M_b{}^d a^{(11)}_d = 0
\end{equation}
where
\begin{gather}\label{eq:M}
M_b{}^d = \chi _{ab}{}^{cd} p ^a p _c 
= -2 \mathcal{L}_F p _a  p ^a \delta _b ^d 
+ 2\mathcal{L} _F p _b p ^d
\\
\nonumber
- 4 \mathcal{L} _{FF} u_b u^d
+ 2 \mathcal{L} _{FG} \big( u _b v^d + v_b u^d \big)
- \mathcal{L} _{GG} v_b v^d
\, .
\end{gather}
Note that $M_b{}^d$ is self-adjoint with respect to the 
spacetime metric, i.e. $M_{ab}=M_{ba}$. This is a consequence of
the above-mentioned fact that the skewon part of the constitutive
tensor vanishes. Also note that the axion part gives no contribution
to (\ref{eq:M}) which is a general result \cite{HehlObukhov2003,Itin2007}. 

From (\ref{eq:M}) we read that $p_d$ is in the kernel of $M_b{}^d$, so
(\ref{eq:MA}) is satisfied by $a^{(11)}_d = \psi \, p _d$ with any
scalar factor $\psi$. However, by (\ref{eq:f11a}) such a potential 
gives a trivial first-order field strength. As we require $f_{cd}^{(11)} \neq 0$,
we need a solution $a^{(11)}_d$ of (\ref{eq:MA}) that is linearly 
independent of $p_d$, i.e., the kernel of $M_d{}^b$ has to be at least two-dimensional.
This is the case if and only if the \emph{adjugate} $A_d{}^b$ of
$M_d{}^b$ (also known as the \emph{classical adjoint}) vanishes,
cf. Itin \cite{Itin2009}. A straight-forward (though tedious)
calculation shows that the adjugate is given by
\begin{gather}\label{eq:adj}
A_b{}^a = - 8 \mathcal{L} _F 
\big( M (p_cp^c)^2
+  N p_cp^c u_du^d 
+ P (u_du^d)^2 \big)
p_b p ^a 
\end{gather}
where
\begin{gather}
M=\mathcal{L}_F^2+2\,\mathcal{L}_F \mathcal{L}_{FG}\,G
-\frac12\,\mathcal{L}_F\mathcal{L}_{GG}\,F -PG^2 \, ,
\label{eq:coeff1}
\\
N=2\,\mathcal{L}_F \mathcal{L}_{FF}
+\frac12\,\mathcal{L}_F\mathcal{L}_{GG} -PF \, ,
\label{eq:coeff2}
\\
P=\mathcal{L}_{FF}\mathcal{L}_{GG}-\mathcal{L}^2_{FG}\,.
\label{eq:coeff3}
\end{gather}
Here we have used the well-known \cite{Plebanski1970} identities 
\begin{equation}\label{eq:FFid}
{}^{*}{\!}F_{ac} F^{bc} = - G \delta _a ^b \, , \quad
F_{ac} F^{bc} -
{}^{*}{\!}F_{ac} {}^{*}{\!} F^{bc}
= F \delta _a^b 
\end{equation}
which imply
\begin{equation}\label{eq:uvorth}
u_cv^c = - G p_cp^c \, , \quad 
u_cu^c-v_cv^c=Fp_cp^c \, .
\end{equation}
By (\ref{eq:adj}), the zeroth-order field equation
(\ref{eq:MA}) admits a solution $a_d^{(11)}$ giving a 
non-trivial field strength if and only if 
\begin{gather}\label{eq:eikonal}
0 = \mathcal{L} _F 
\big( M (p_cp^c)^2
+  N p_cp^c u_du^d 
+ P (u_du^d)^2 \big)
\, .
\end{gather}
This is the \emph{eikonal equation}. It is a first-order partial
differential equation for the function $S$. Each solution to this
equation determines a family of light rays, in the same way as in
Hamiltonian mechanics each solution to the Hamilton-Jacobi equation
determines a family of trajectories, see the next subsection.
If viewed as an algebraic condition on the covector $p_a$, 
(\ref{eq:eikonal}) is known as the \emph{dispersion relation}, 
as the \emph{characteristic equation} or as the \emph{Fresnel 
equation}. 

From now on we reqire $\mathcal{L}_F \neq 0$ because
otherwise the eikonal equation is an identity, so there 
is no well-defined notion of rays. If in addition $M \neq 0$, 
(\ref{eq:eikonal}) factorizes according to 
\begin{equation}\label{eq:eikonalfact}
\big( \tilde{g}{}^{bc}_+ p_b p_c\big)  
\big( \tilde{g}{}^{de}_- p_d p_e \big)  
= 0
\end{equation}
where
\begin{gather}
\nonumber
\tilde{g}{}^{bc}_{\pm} 
= g^{bc} + \sigma _{\pm} F^{bd}F^c{}_d
\\
= \big( 1 + \sigma _{\pm} F \big) g^{bc} 
+ \sigma _{\pm} {}^{*} {\!} F^{bd} {}^{*} {\!} F^c{}_d
\label{eq:optmet}
\end{gather}
and 
\begin{equation}\label{eq:sigma}
\sigma _{\pm} =
\dfrac{N}{2M} \pm
\sqrt{\dfrac{N^2}{4M^2}-\dfrac{P}{M} \,} \, .
\end{equation}
$\tilde{g}{}^{bc}_{+}$ and $\tilde{g}{}^{bc}_{-}$ are 
known as the \emph{optical metrics}. Note that 
$\sigma_{\pm}$ is always real because $N^2-4MP$ can
be rewritten as the sum of two squares,
\begin{gather}\label{eq:sigmareal}
N^2-4MP= 
\\
\nonumber
\Big( \mathcal{L}_F \mathcal{L}_{GG}
 - 
N \Big) ^2
+ 4 \Big( \mathcal{L}_F \mathcal{L}_{FG}-PG \Big) ^2 \, .
\end{gather}

The determinant of $\tilde{g}{}^{cd}_{\pm}$ is
\begin{equation}\label{eq:nondeg}
\mathrm{det} \big( \tilde{g}{}^{cd}_{\pm} \big) =
\big( 1 + \sigma _{\pm} F - \sigma _{\pm}^2 G^2 \big)^2 
\mathrm{det} \big( g^{cd} \big)
\, .
\end{equation}
As $(g^{cd})$ is of Lorentzian signature, the right-hand side of
(\ref{eq:nondeg}) is either zero or negative. This demonstrates 
that the optical metrics
are either degenerate or Lorentzian (i.e., of signature $(-+++)$
or $(---+)$\big), as was already observed by Obukhov and 
Rubilar \cite{ObukhovRubilar2002}. If the determinant is non-zero,
the covariant components  of the optical metrics are
\begin{gather}
\nonumber
\big( \tilde{g}{}^{-1} \big)^{\pm}_{cd} = 
\dfrac{
g_{cd} - \sigma _{\pm} {}^{*} {\!} F_c{}^{b \,} {}^{*} {\!} F_{db}
}{
1+ \sigma _{\pm} F - \sigma _{\pm}^2 G^2
}
\\
=
\dfrac{
\big( 1 + \sigma _{\pm} F \big) g_{cd} - \sigma _{\pm} F_c{}^b F_{db}
}{
1+ \sigma _{\pm} F - \sigma _{\pm}^2 G^2
}
\, .
\label{eq:covoptmet}
\end{gather}
Indeed, with the help of the identities (\ref{eq:FFid})
it is easy to check that (\ref{eq:optmet}) and
(\ref{eq:covoptmet}) imply $\big( \tilde{g}{}^{-1} \big)
^{\pm}_{ac} \tilde{g}{}^{cb}_{\pm} = \delta _a^b$.

If $M=0$, the eikonal equation factorizes as well, but we
will not consider this case because it shows some pathologies,
see \cite{SchellstedePerlickLaemmerzahl2016}. We restrict for 
the rest of the paper to background fields for which 
$\mathcal{L}_F \neq0$, $M \neq 0$ and
$(1+\sigma _{\pm} F - \sigma _{\pm}^2 G^2 ) \neq 0$ so that we
have two optical metrics of Lorentzian signature. Then the eikonal
equation is of the form (\ref{eq:eikonalfact}), i.e., it requires
$p_a = \nabla _a S$ to be a null covector of at least one of the 
two optical metrics. This is true if and only if 
\begin{equation}\label{eq:eikonsigma}
p_ap^a + \sigma _{\pm} u_au^a =0
\end{equation}
holds with at least one of the two signs where $\sigma _{\pm}$ is
given by (\ref{eq:sigma}). We refer to the two equations 
(\ref{eq:eikonsigma}) with $p_a = \nabla _a S$ as to the
two \emph{partial eikonal equations}. 

We end this section with two useful results.
\begin{proposition}\label{prop:nobiref}
Let $\sigma$ be one of the two solutions, $\sigma = \sigma _+$ 
or $\sigma = \sigma _-$, to \emph{(\ref{eq:sigma})}. Then the following 
conditions are mutually equivalent:
\begin{itemize}
\item[\emph{(a)}]
$N^2=4MP$, i.e., the two optical metrics coincide, $\sigma = \dfrac{N}{2M}$ .
\item[\emph{(b)}]
$\mathcal{L}_F \mathcal{L}_{GG} = N$ and 
$\mathcal{L}_F \mathcal{L}_{FG}=PG$.
\item[\emph{(c)}]
$D M = \mathcal{L}_F^2$,
$D N = 2 \mathcal{L}_F ^2 \sigma$ and
$D P = \mathcal{L}_F^2 \sigma ^2$.
\item[\emph{(d)}]
$2 D \mathcal{L}_{FF} = 
\mathcal{L}_F \sigma ( 1+F \sigma )$,
$D \mathcal{L}_{GG} = 
2 \mathcal{L}_F \sigma$ and
$D \mathcal{L}_{FG} = 
\mathcal{L}_F  G \sigma ^2$.
\end{itemize}
In \emph{(c)} and \emph{(d)}, $D = 1 + F \sigma - G^2 \sigma ^2$.
\end{proposition}
\begin{proof}
(a) $\Leftrightarrow$ (b) is obvious from (\ref{eq:sigmareal}). 
We now assume that one, and thus also the other, of these conditions is true.
Then we find from (a) that
\begin{equation}\label{eq:NPsigma}
N=2 M \sigma \, , \quad P = M \sigma ^2 
\end{equation}
and from inserting (b) into (\ref{eq:coeff1}) that
\begin{equation}\label{eq:Msigma}
M= \mathcal{L}_F^2-PG^2- \dfrac{NF}{2} \, 
\end{equation}
(\ref{eq:NPsigma}) and (\ref{eq:Msigma}) demonstrate that then (c) 
is true. Conversely, (c) obviously implies (a), so we have proven
that (a), (b), and (c) are mutually equivalent. Finally, we 
observe that (a) and (c) together with (\ref{eq:coeff2}) imply 
(d) and that (d), if inserted into (\ref{eq:coeff1}), (\ref{eq:coeff2}) 
and (\ref{eq:coeff3}), implies (a), so all four conditions are indeed 
mutually equivalent.
\end{proof}
\begin{proposition}\label{prop:eigen}
Assume that $p_a = \nabla _a S$ is a solution to the eikonal equation
$p_a p^a + \sigma u_au^a =0$ with $\sigma = \sigma _+$ or 
$\sigma = \sigma _-$. Then the eigenvalues of the matrix 
$\big( M_b{}^d \big)$ are 
$\lambda _1 = \lambda _2 =0$ and 
\begin{equation}\label{eq:lambda4}
\lambda _3 =
2 \mathcal{L}_F  \sigma u_au^a \, ,
\end{equation}
\begin{equation}\label{eq:lambda3}
\lambda _4 =
\big( 4 \mathcal{L} _F \sigma - 4 \mathcal{L}_{FF}
+ 4 \mathcal{L}_{FG} G \sigma -\mathcal{L}_{GG} (1+F \sigma ) \big) u_au^a
\end{equation}
\end{proposition}
\begin{proof}
By assumption, zero is a double-eigenvalue of the matrix (\ref{eq:M}). 
Then the remaining two eigenvalues  $\lambda _3$ and $\lambda _4$ can 
be determined in the following way. The formulas for the trace of a 
matrix and for the trace of the square of a matrix in terms of its 
eigenvalues yield
\begin{equation}\label{eq:trc1}
M_b{}^b = \lambda _3 + \lambda _4 \, ,
\end{equation}
\begin{equation}\label{eq:trc2}
M_b{}^dM_d{}^b = \lambda _3^2 + \lambda _4^2   \, .
\end{equation}
Upon calculating the traces with the help of (\ref{eq:puvorth}), 
solving (\ref{eq:trc1}) and (\ref{eq:trc2}) for the eigenvalues
results in the given expressions for $\lambda _3$ and $\lambda _4$.
\end{proof}
%

\subsection{Hamiltonian for rays and transport vector fields}
\label{subsec:rays}

We say that $S$ is a solution to the eikonal equation 
of multiplicity two if $p_a = \nabla _a S$ satisfies 
the equation (\ref{eq:eikonsigma}) with both signs, 
and we say that it is a solution of multiplicity one if 
(\ref{eq:eikonsigma}) holds with one sign but not with the
other. The multiplicity may change from point to point.

Each of the two partial eikonal equations has the form
of the Hamilton-Jacobi equation, $H(x,\nabla S)=0$,
with the Hamiltonian
\begin{equation}\label{eq:Hpart}
H_{\pm} (x,p ) = \dfrac{1}{2} \tilde{g}{}^{bc}_{\pm} (x) p_b p_c
\, .
\end{equation}
The solutions to Hamilton's equations
\begin{equation}\label{eq:Hampart}
\dot{x}{}^a \! = \! \dfrac{\partial H_{\pm}(x,p)}{\partial p_a}
\, , \;
\dot{p}{}_a \! = \! -  
\dfrac{\partial H_{\pm}(x,p)}{\partial x^a}
\, , \;
H(x,p) \! = \! 0 
\end{equation}
are known as the \emph{bicharacteristic curves} or as the
\emph{rays}. They are
the null geodesics of the optical metric.
Every solution $S$ to the eikonal equation 
is associated with a congruence of
rays whose tangent vector field is given by 
\begin{equation}\label{eq:Ka}
K^b_{\pm} (x) = \dfrac{\partial H_{\pm} (x,p)}{\partial p_b}
\Big| _{p = \nabla S(x)} 
=
\tilde{g}{}^{bc}_{\pm}(x)\nabla _c S(x) 
\, ,
\end{equation}
i.e.,
\begin{equation}\label{eq:Ksigma}
K^b_{\pm} = p^b - \sigma _{\pm} F^{bc} u_c \, .
\end{equation}
This vector field is known as the \emph{transport vector
field} associated with the solution $S$ of the eikonal
equation. For solutions of multiplicity two, we have 
two transport vector fields $K^b_+$ and $K^b_-$. However, 
they are always proportional to each other so that the
rays (as unparametrized curves) are uniquely determined.
We will prove this in the next section. Note that the non-degeneracy
of the optical metric implies that the transport
vector field cannot have zeros if we assume that $p_a = \nabla _a S$
has no zeros (as required for an eikonal function of an approximately
plane wave), i.e, that ``rays cannot stand still''.

The following proposition establishes a property of the transport vector 
field that will be crucial for the next section.
\begin{proposition}\label{prop:aux1}
Assume that $p_a = \nabla _a S$ satisfies the eikonal equation
$p_ap^a + \sigma u_au^a=0$ where $\sigma = \sigma _+$ or 
$\sigma = \sigma _-$. Let $\tilde{g}{}^{ab}=g^{ab} + \sigma F^{ac}F^b{}_c$
and $K^a = \tilde{g}{}^{ab}p_b$. Then 
\begin{gather}
\label{eq:optpuv1}
\tilde{g}{}^{cd} p_c u_d=
\tilde{g}{}^{cd} p_c v_d = 0 \, , \quad
\tilde{g}{}^{cd} u_c v_d= 0 \, ,
\\
\label{eq:optpuv2}
\tilde{g}{}^{cd} u_c u_d= 
\tilde{g}{}^{cd} v_c v_d= u^cu_c(1+\sigma F - \sigma ^2 G^2) \, .
\end{gather} 
As a consequence, the transport vector field satisfies 
\begin{equation}\label{eq:Kpuv}
K^ap_a = K^a u_a = K^a v_a = 0 \, .
\end{equation}
\end{proposition}
\begin{proof}
This can be verified in a straight-forward manner with the help
of the identities (\ref{eq:FFid}).
\end{proof}

\subsection{Polarization condition}
\label{subsec:polcon}

If we fix a solution $p_a=\nabla _aS$ to the eikonal equation
$p_ap^a+ \sigma u_au^a = 0$ with $\sigma = \sigma _+$ or $\sigma = 
\sigma _-$, the zeroth-order field equation (\ref{eq:MA}) gives an 
algebraic restriction on $a^{(11)}_b$. This is the zeroth-order
polarization condition. In this section we investigate to what 
extent the polarization condition fixes the allowed values for 
$a^{(11)}_b$ and, thereby, for the lowest-order field-strength 
amplitude $f_{cd}^{(11)}$.

Thereby we have to distinguish solutions of multiplicity two from
solutions of multiplicity one. Clearly, if the two optical metrics 
coincide, $\sigma _+ = \sigma _-$, every solution is of multiplicity 
two. In a background field with $\sigma _+ \neq \sigma _-$, a 
solution is of multiplicity two if and only if $u_au^a =0$.
In this case $p_a$ is a \emph{principal null covector}, i.e.,
a covector with $p_ap^a=0$ for which $u_a$ and $v_a$ are
multiples of $p_a$. In the following proposition we determine
the general form of the matrix $M_b{}^d$ for this special case.
For more details on principal null solutions to the eikonal
equation we refer to Abalos et al. \cite{Abalosetal2015} where 
also pictures of the cones of the optical metrics can be found.

\begin{proposition}\label{prop:aux2a}
Assume that $p_a =\nabla _a S$ satisfies $p_ap^a=0$ and $u_au^a=0$.
Then $p_a$ is a solution of multiplicity two to the eikonal equation.
The covectors $u_a$ and $v_a$ are multiples of $p_a$,
\begin{equation}\label{eq:princnull}  
u_c = F_{c}{}^ap_a = \mu p_c \, , \quad
v_c = {}^{*} {\!} F_{c}{}^ap_a = \nu p_c \, ,
\end{equation}
where the coefficients $\mu$ and $\nu$ satisfy
\begin{gather}
\mu ^2 = - \dfrac{F}{2} + \sqrt{ \dfrac{F^2}{4}+G^2} 
\, , \quad
\nu ^2 = \dfrac{F}{2} + \sqrt{ \dfrac{F^2}{4}+G^2} \, ,
\nonumber
\\
\mu \nu  =  -G  \, .
\label{eq:munu}
\end{gather}
The transport vector fields are proportional to $p^a$,
\begin{equation}\label{eq:Kprinc}
K^a _{\pm} =  \xi _{\pm} \, p^a \, ,
\end{equation}
where 
\begin{equation}\label{eq:xi}
\xi _{\pm}= 
1 - \sigma _{\pm} \mu ^2 \, .
\end{equation}
The matrix $M_b{}^d$ reduces to
\begin{equation}\label{eq:Mprinc}
M_d{}^b= \Big( 2 \mathcal{L}_F - 4 \mathcal{L}_{FF} \mu ^2
+ 4  \mathcal{L}_{FG} \mu \nu
- \mathcal{L}_{GG} \nu ^2 \Big) p_b p^d \, .
\end{equation}
\end{proposition}
\begin{proof}
If $p_ap^a=0$ and $u_au^a=0$, (\ref{eq:eikonsigma}) is trivially satisfied
with both signs, i.e., the covector $p_a$ is lightlike with respect to both 
optical metrics. Moreover, we read from (\ref{eq:optpuv1}) and (\ref{eq:optpuv2}) 
that with respect to either of the two optical metrics the covectors $u_a$ and 
$v_a$ are orthogonal to $p_a$ and lightlike. As two lightlike vectors
are orthogonal with respect to a Lorentzian metric if and only if they
are linearly dependent, this proves that (\ref{eq:princnull}) has to hold with
some coefficients $\mu$ and $\nu$. Then (\ref{eq:munu}) follows from (\ref{eq:FFid}). 
Inserting (\ref{eq:princnull}) into (\ref{eq:Ksigma}) and (\ref{eq:M}), respectively,
yields (\ref{eq:Kprinc}) and (\ref{eq:Mprinc}).
\end{proof}
Recall that the eikonal equation requires the kernel of $M_b{}^d$ to be at
least two-dimensional. Proposition (\ref{prop:aux2a}) implies that the kernel
is even three-dimensional if $u_au^a=0$. We will now consider the case
$u_au^a \neq 0$.
\begin{proposition}\label{prop:aux2b}
Assume that $p_a = \nabla _a S$ is a solution to one of the two eikonal
equations, $p_ap^a + \sigma u_au^a =0$ where $\sigma$ stands for $\sigma _+$ or 
for $\sigma _-$. Let $\tilde{g}^{ab} = g^{ab} + \sigma F^{ac}F^b{}_c$ be the 
corresponding optical metric and $K^a=\tilde{g}^{ab}p_b$ be the corresponding
transport vector field. If $u_au^a \neq 0$, the three covectors $p_a$, $u_a$ 
and $v_a$ are linearly independent. They span the orthocomplement of $p_a$
with respect to $\tilde{g}^{ab}$. The kernel of the matrix $M_b{}^d$ consists
of all covectors 
\begin{equation}\label{eq:alphabetagamma}
a^{(11)}_b = \alpha u_b + \beta v_b + \gamma p_b 
\end{equation}
where $\gamma$ is arbitrary and $\alpha$ and $\beta$ satisfy
\begin{equation}\label{eq:alphabeta}
\begin{pmatrix}
m_1{}^1 & m_1{}^2 \\ m_2{}^1 & m_2{}^2
\end{pmatrix}
\begin{pmatrix}
\alpha
\\
\beta
\end{pmatrix}
=
\begin{pmatrix}
0
\\
0
\end{pmatrix}
\end{equation}
where
\begin{gather}
\begin{pmatrix}
\nonumber
m_1{}^1 & m_1{}^2 \\[0.1cm] m_2{}^1 & m_2{}^2
\end{pmatrix}
=
2 \mathcal{L}_F \sigma
\begin{pmatrix}
1 & 0 \\[0.1cm] 0 & 1
\end{pmatrix}
\\[0.2cm]
-
\begin{pmatrix}
-4 \mathcal{L}_{FF} & 2\mathcal{L}_{FG}
\\[0.1cm]
2 \mathcal{L}_{FG} & - \mathcal{L}_{GG}
\end{pmatrix}
\begin{pmatrix}
1 & G \sigma
\\[0.1cm]
\; G \sigma \; & 1+F \sigma
\end{pmatrix}
\, .
\label{eq:m}
\end{gather}
The kernel is three-dimensional if and only if $p_a = \nabla _aS$ is 
a solution of multiplicity two. The kernel is then spanned by $p_d$, $u_d$
and $v_d$, i.e., it coincides with the orthocomplement of $p_b$ with
respect to the optical metric.
\end{proposition}
\begin{proof}
Our assumption that $u_au^a \neq 0$ implies, by (\ref{eq:optpuv1}) and 
(\ref{eq:optpuv2}) together with $\tilde{g}^{ab}p_ap_b=0$, that $p_a$, $u_a$
and $v_a$ are linearly independent and that they span the 
$\tilde{g}^{ab}$-orthocomplement of $p_a$. After normalizing $u_a$ and $v_a$
with the help of (\ref{eq:optpuv2}) we may complement these three covectors
to a Newman-Penrose tetrad by choosing a covector $w_a$ with
\begin{gather}
\tilde{g}^{ab}w_ap_b=1 \, , \quad
\tilde{g}^{ab}w_aw_b= 0 \, ,
\nonumber
\\
\tilde{g}^{ab}w_au_b= 0 \, , \quad 
\tilde{g}^{ab}w_av_b=0 \, .
\label{eq:w}
\end{gather}
From (\ref{eq:M}) we calculate with the help of (\ref{eq:uvorth})  
\begin{gather}
M_b{}^dw_d= 2 \mathcal{L}_F \sigma u_au^a w_b 
\nonumber
\\
+ 2 \mathcal{L}_F  \big( 1+ \sigma F^{fg} F^e{}_g w_fp_e \big) p_b 
\, ,
\label{eq:Mw}
\end{gather}
\begin{equation}\label{eq:Mp}
M_b{}^dp_d= 0 \, ,
\end{equation}
\begin{gather}
M_b{}^du_d= u_au^a  
\big(2 \mathcal{L} _F \sigma - 4 \mathcal{L}_{FF}
+ 2 \mathcal{L}_{FG} G \sigma \big) u_b 
\nonumber
\\
+
u_au^a \big(2 \mathcal{L} _{FG}  - \mathcal{L}_{GG}  G \sigma \big) v_b 
\, ,
\label{eq:Mu}
\end{gather}
\begin{gather}
M_b{}^dv_d= u_au^a 
\big(2 \mathcal{L} _{FG} (1+F \sigma)- 4 \mathcal{L}_{GG}  G \sigma \big) u_b 
\nonumber
\\
+ 
u_au^a
\big(2 \mathcal{L} _F \sigma - \mathcal{L}_{GG} (1+F \sigma )
+ 2 \mathcal{L}_{FG} G \sigma \big) v_b 
\, .
\label{eq:Mv}
\end{gather}
The first two equations (\ref{eq:Mw}) and (\ref{eq:Mp}) demonstrate that 
$M_b{}^d$ leaves the two-space spanned by $w_d$ and $p_d$ invariant and 
that it has a one-dimensional kernel on this two-space. The last statement 
follows from the fact that $w_d$ is not in the kernel: It is mapped onto a 
covector $M_b{}^dw_d$ that is non-zero if $\sigma =0$ (because then it is a 
non-zero multiple of $p_b$) and also if $\sigma \neq 0$ (because then it has 
a non-zero component in the direction of $w_b$). The other two equations
(\ref{eq:Mu}) and (\ref{eq:Mv}) demonstrate that the two-space spanned by
$u_d$ and $v_d$ is left invariant as well. On this two-space the matrix 
$M_b{}^d$ must have a one-dimensional or two-dimensional kernel because
the eikonal equation requires that the kernel of the full matrix $M_b{}^d$ 
is at least two-dimensional. By (\ref{eq:Mu}) and (\ref{eq:Mv}), a
covector $\alpha u_b + \beta v_b$ is in the kernel if and only if 
(\ref{eq:alphabeta}) holds with (\ref{eq:m}). The determinant of the 
matrix (\ref{eq:m}) vanishes as a consequence of the eikonal equation.
Clearly, a ($2 \times 2$)-matrix has a two-dimensional kernel if and 
only if it is the zero matrix. The matrix (\ref{eq:m}) is the zero matrix 
if and only if the symmetric matrix 
\begin{gather}
\nonumber
\dfrac{1}{u_au^a D}
\begin{pmatrix}
m_1{}^1 & m_1{}^2 \\[0.1cm] m_2{}^1 & m_2{}^2
\end{pmatrix}
\begin{pmatrix}
1 + F \sigma & -G \sigma \\[0.1cm] - G \sigma & 1
\end{pmatrix} =
\\[0.2cm]
\dfrac{2 \mathcal{L}_F \sigma}{D}
\begin{pmatrix}
1+ F \sigma & - G \sigma \\[0.1cm] - G \sigma & 1 
\end{pmatrix}
+
\begin{pmatrix}
- 4 \mathcal{L}_{FF} & 2 \mathcal{L}_{FG} \\[0.1cm]
2 \mathcal{L}_{FG} & - \mathcal{L}_{GG}
\end{pmatrix}
\label{eq:msym}
\end{gather}
is the zero matrix, where $D=1+F \sigma -G^2 \sigma ^2$. By comparison 
with part (d) of Proposition \ref{prop:nobiref} we see that this is the 
case if and only if the two optical metrics coincide. As we assume that 
$u_au^a \neq 0$ this is true if and only if $p_a=\nabla _a S$ 
is a solution of multiplicity two.
\end{proof}
With these results at hand it is now easy to evaluate the polarization
condition. We do this first for solutions of multiplicity two.
\begin{proposition}\label{prop:polconmult2}
Let $p_a= \nabla _a S$ be a solution of multiplicity two
to the eikonal equation, i.e.
$p_ap^a + \sigma _+ u_au^a = 0$ and
$p_ap^a + \sigma _- u_au^a = 0$. Then the 
two transport vector fields $K^a_+ = p^a - \sigma _+ F^{ab}u_b$
and $K^a_- = p^a - \sigma _-  F^{ab}u_b$ are linearly 
dependent. The polarization condition $M_b{}^d a^{(11)}_d=0$ is 
equivalent to $K^d_{\pm} a^{(11)}_d =0$ (which holds with one sign if and
only if it holds with the other sign), i.e., it restricts $a^{(11)}_d$ 
to a three-dimensional subspace which contains $p_a$.
\end{proposition}
\begin{proof}
If $u_au^a=0$, this follows from Proposition \ref{prop:aux2a}. If
$u_au^a \neq 0$ it follows from Proposition \ref{prop:aux2b}.
\end{proof}
We now prove the analogous statement for solutions of multiplicity one.
\begin{proposition}\label{prop:polconmult1}
Let $p_a= \nabla _a S$ be a solution of multiplicity one
to the eikonal equation, i.e.
$p_ap^a + \sigma u_au^a = 0$ with $\sigma = \sigma _+$ or
$\sigma = \sigma _-$ but not with both. Then the polarization 
condition $M_b{}^da^{(11)}_d=0$ is true if and only if $a_b^{(11)}
= \alpha u_b + \beta v_b + \gamma p_b$ where $\alpha$ and $\beta$
satisfy \emph{(\ref{eq:alphabeta})} with \emph{(\ref{eq:m})}. This condition
restricts $a^{(11)}_b$ to a two-dimensional subspace that contains
$p_b$.
\end{proposition}
\begin{proof}
This is an immediate consequence of Proposition \ref{prop:aux2b}).
\end{proof}
We summarize the results of this section in the following way. For every 
solution $p_a=\nabla _aS$ to the eikonal equation the polarization condition 
requires that $a^{(11)}_b$ satisfies $K^ba^{(11)}_b =0$ where $K^b$ is the
corresponding transport vector field. This may be interpreted as a transversality
condition. For a solution of multiplicity two there is no additional restriction,
i.e., $a^{(11)}_b$ is confined to a three-dimensional subspace that contains
$p_b$. By contrast, for a solution of multiplicity one the polarization
condition restricts $a^{(11)}_b$ to a two-dimensional space that contains $p_b$,
i.e., it fixes the polarization plane (the plane spanned by $a_b^{(11)}$ and $p_b$)
uniquely.

\section{Evaluation of the first-order field  equation}\label{sec:firstorder}

We now turn to the first-order field equation which gives us the
two conditions (\ref{eq:Max2first1a}) and (\ref{eq:Max2first2a}). We can write
them, in a slightly more compact form, as
\begin{equation}\label{eq:Max2first1b}
0 = \nabla ^a \Big( \chi _{ab}{}^{cd} p_c \, a_d^{(11)} \Big)
+  p ^a  \, \chi _{ab}{}^{cd} \nabla _c a_d^{(11)}
+ i \, M_b{}^d \, a_d ^{(21)} \, ,
\end{equation}
\begin{equation}\label{eq:Max2first2b}
0 = M_b{}^d \, a_d ^{(22)}
- \psi _{ab}{}^{cdef} p ^a  \, p _c \, a_d ^{(11)} a_f^{(11)} \, . 
\end{equation}
We want to determine what kind of information these equations give us
on the polarization plane spanned by $a_b^{(11)}$ and $p_b$.

We know from the preceding section that for a solution of multiplicity one this
plane is already uniquely fixed at the zeroth-order level, so the first-order 
equations cannot give us any additional information on this plane. One just has to 
check for consistency, i.e., one has to verify that the sum of the first two terms 
in (\ref{eq:Max2first1b}) is in the image space of $M_b{}^d$ and that the second term 
in (\ref{eq:Max2first2b}) is in the image space of $\psi _{ab}{}^{cdef}$. Then 
(\ref{eq:Max2first1b}) and (\ref{eq:Max2first2b}) give us polarization conditions on
$a_d^{(21)}$ and $a_d^{(22)}$. We have 
already emphasized that (\ref{eq:Max2first2b}) is not in general satisfied by 
$a_d^{(22)} = 0$, i.e., that frequency doubling has to be taken into account if 
the field equation should hold at first order, and that at the next order in general
also a non-zero $a^{(20)}$ is needed.

As in this paper we will be satisfied with determining the potential up to
first order, there is nothing else to be done for solutions of multiplicity one. 
Therefore, in the following we will restrict ourselves to solutions of 
multiplicity two. We know from the preceding section that then $a_d^{(11)}$ 
is restricted at the zeroth-order level only by the condition $K^d a_d^{(11)} =0$. 
This condition restricts the polarization plane to a three-dimensional space, i.e.,
it still allows the polarization plane to arbitrarily rotate along 
a ray. We will now demonstrate that the first-order equation (\ref{eq:Max2first1b})
gives us a transport law which uniquely determines the polarization plane
along a ray if it is given at one point of this ray. We will consider 
first solutions of multiplicity two with $u_a u^a =0$ and then with
$u_au^a \neq 0$.

\subsection{Transport equation in the case $\boldsymbol{u_au^a=0}$}\label{sec:transport1}

For a solution of multiplicity two with $u_au^a=0$ the rays are lightlike 
geodesics not only with respect to each of the two optical metrics but also 
with respect to the spacetime metric. (The affine parametrizations are in 
general different.) For such a solution we have $u_a$ and $v_a$ parallel to 
$p_a$ and the matrix $M_b{}^d$ projects onto the line spanned by $p_b$, 
recall Proposition \ref{prop:aux2a}. As a consequence, (\ref{eq:Max2first1b}) 
reduces to
\begin{gather}
4 \, \mathcal{L} _F p^a \nabla _a a_b^{(11)} 
+ 2 \nabla _a \big( \mathcal{L} _F p^a \big) a_b^{(11)}
\nonumber
\\
- \nabla ^a \mathcal{L} _G \, \varepsilon _{ab}{}^{cd} p_c a_d^{(11)}
= \psi \, p_b
\label{eq:transprinc}
\end{gather}
where $\psi$ is an undetermined scalar function.
Recall from Proposition \ref{prop:aux2a} that in the case at hand the two
transport vector fields $K^a_+$ and $K^a_-$ are multiples of $p^a$, i.e., 
that $p^a$ is tangent to the rays. Therefore, (\ref{eq:transprinc}) gives us
a first-order ordinary differential equation for $a_b^{(11)}$ along each
ray. As $\psi$ is arbitrary, for each initial condition this differential 
equation has a solution that is unique up to a multiple of $p_b$. In other
words, (\ref{eq:transprinc}) gives us a unique transport law for the 
polarization plane.

If $\nabla ^a \mathcal{L} _G \, \varepsilon _{ab}{}^{cd} p_c a_d^{(11)} =0$, we read 
from (\ref{eq:transprinc}) that the polarization plane is parallel with respect to 
the transport law defined by the Levi-Civita derivative of the spacetime metric, as 
it is in the standard Maxwell vacuum theory; in general, however, in a theory of the 
Pleba{\'n}ski class a background field with non-constant $\mathcal{L} _G$ produces a 
rotation of the polarization plane. This gives us a new experimental test of this type
of theories in situations where the rays behave as in the standard vacuum Maxwell 
theory but the polarization plane does not. We will exemplify this with the Born-Infeld 
theory in the next section.

\subsection{Transport equation in the case $\boldsymbol{u_au^a \neq 0}$}\label{sec:transport2}

We now consider a solution of multiplicity two with $u_au^a \neq 0$. For such
solutions we know from Proposition \ref{prop:aux2b} that the matrix $M_b{}^d$
has a three-dimensional kernel spanned by $p_b$, $u_b$ and $v_b$, i.e., that
$a_d^{(11)}$ is of the form
\begin{equation}\label{eq:albega}
a_d^{(11)} = \alpha u_d + \beta v_d + \gamma p_d \, .
\end{equation}
By the same token, as the matrix $M_b{}^d$ is self-adjoint with respect to the 
spacetime metric, (\ref{eq:Max2first1b}) is true with some $a_d^{(21)}$ if and 
only if the equation
\begin{equation}\label{transportx}
0 = z^b \Big\{\nabla ^a \Big( \chi _{ab}{}^{cd} p_c \, a_d^{(11)} \Big)
+  p ^a  \, \chi _{ab}{}^{cd} \nabla _c a_d^{(11)} \Big\} 
\end{equation}
is true for $z^b=p^b$, $z^b=u^b$ and $z^b=v^b$. It is easy to check that for 
$z^b=p^b$ the equation is identically satisfied, for all $a_d^{(11)}$ of the 
form (\ref{eq:albega}). Therefore, we only have to consider it for $z^b=u^b$ and 
$z^b=v^b$. 

To that end, we recall that a solution of multiplicity two with $u_au^a \neq 0$ 
exists only if the two optical metrics coincide. In the following we write 
$\sigma$ for $\sigma _+ = \sigma _-$ and  $K^a$ for $K^a_+ = K^a_-$. If we express 
$\mathcal{L}_{FF}$, $\mathcal{L}_{FG}$ and $\mathcal{L}_{GG}$ with the help of 
part (d) of Proposition \ref{prop:nobiref}, we see that $\chi _{ab}{}^{cd}$ can 
be written as
\begin{gather}
\chi _{ab}{}^{cd} = \mathcal{L} _G \varepsilon _{ab}{}^{cd}
- 2 \, \mathcal{L}_F \big( \delta _a^c \delta _b^d - \delta _a^d \delta _b^c \big)
\nonumber
\\
- \dfrac{2 \, \mathcal{L}_F \sigma}{D}
F_{ab} \big( (1+ \sigma F ) F^{cd} - \sigma G {}^{*} {\!} F^{cd} \big)
\nonumber
\\
+ \dfrac{2 \, \mathcal{L}_F \sigma}{D}
{}^{*} {\!} F_{ab}  \big( \sigma G F^{cd} - {}^{*} {\!} F^{cd} \big) \Big)
\label{eq:chimulttwo}
\end{gather}
where $D=1+\sigma F - \sigma ^2 G^2$.

If we insert this expression and (\ref{eq:albega}) into (\ref{transportx}) with 
$x^b=u^b$ and with $x^b=v^b$, we get after some lengthy algebra the two equations
\begin{equation}\label{eq:transport1}
4 \, \mathcal{L}_F u_bu^b K^a \nabla _a \alpha 
= a \, \alpha + b \, \beta \, ,
\end{equation}
\begin{equation}\label{eq:transport2}
4 \, \mathcal{L}_F u_bu^b K^a \nabla _a \beta 
= a \, \beta - b \, \alpha \, ,
\end{equation}
where
\begin{equation}\label{eq:defa}
a = \, - \, 2 \, \nabla _a \big( \mathcal{L}_F u_cu^c K^a \big)   \, ,
\end{equation} 
\begin{gather}
b = \nabla ^a \mathcal{L} _G \, \varepsilon _{abcd} p^bv^cu^d
\nonumber
\\
+ 2 \, \mathcal{L}_F p^b \big( p^ap^x-p_ep^e g^{ac} \big)
\big( {}^{*} {\!} F_a{}^d \nabla _cF_{db}
+F_a{}^d \nabla _c {}^{*} {\!} F_{db} \big) \, .
\label{eq:defb}
\end{gather} 
These equations determine the change of $\alpha$ and $\beta$ and, thus, of the
polarization plane, along each ray. In particular, $b$ determines the rotation 
of the polarization plane with respect to the basis covectors $u_b$ and $v_b$ 
which are orthogonal to each other, but not parallelly transported, with respect 
to the optical metric.


\section{Example: Born-Infeld theory}\label{sec:BI}
As an example, we consider the transport law for the polarization
plane in the Born-Infeld theory \cite{BornInfeld1934} where the 
Lagrangian is given by (\ref{eq:LBI}). In this case, for
any background field, the two optical metrics coincide,
\begin{equation}\label{eq:sigmaBI}
\sigma _+ = \sigma _- =  \dfrac{-1}{b_0^2+F} \, ,
\end{equation}
\begin{equation}\label{eq:gBI}
\tilde{g}{}_{+}^{ab} = \tilde{g}{}_{-}^{ab} = g^{ab} - \dfrac{F^{ac}F^b{}_c}{b_0^2 +F}
\end{equation} 
so there is no birefringence and every solution to the eikonal equation
is a solution of multiplicity two.

We assume that the underlying spacetime is the Minkowski spacetime with 
standard inertial coordinates $(x^0=ct,x^1,x^2,x^3)$, i.e., $g_{ab}
= \eta _{ab}$ where $(\eta _{ab} ) = \mathrm{diag} (-1,1,1,1)$. As the
background field we choose the superposition of a time-dependent electric field
and a constant magnetic field, both in $x^3$ direction; so the only 
non-vanishing components of the field strength tensor are
\begin{equation}\label{eq:BIF}
F_{03} = - F_{30} = \dfrac{E(t)}{c} 
\, , \quad
F_{12}= - F_{21} = B_0 \, .
\end{equation}
Note that the homogeneous Maxwell
equation is indeed satisfied, $\varepsilon ^{abcd} \nabla _b F_{cd} =0$. For this 
electromagnetic field, 
\begin{equation}\label{eq:BIp}
p_a = \nabla _a S \, , \quad S(x^0,x^1,x^2,x^3) = \dfrac{1}{c} \big( x^3-x^0 \big)
= \dfrac{x^3}{c} - t
\end{equation}
is a principal null covector field, $p_bp^b=0$ and
\begin{equation}\label{eq:BIu}
u_a=F_{ab}p^b = - E  ( t ) \, p_a 
\, , 
\end{equation}
\begin{equation}\label{eq:BIu}
v_a= {}^{*} {\!} F_{ab}p^b = - c B_0 p_a \, , 
\end{equation}
so the eikonal equation is satisfied with $u_au^a=0$. The transport vector 
field $K^a$ is  proportional to $p^a$, i.e., the rays are straight lightlike 
lines in $x^3$ direction. Note that they are lightlike not only with respect to
the spacetime metric; they are lightlike geodesics also with respect to 
the optical metric. Whereas $K^a$ is adapted to an affine parameter with
respect to the optical metric, $p^a$ is adapted to the parametrization with
the time coordinate $t$ which is an affine parameter with respect to the
spacetime metric. 

The amplitude $a_b^{(11)}$ must be orthogonal to $p_c$ with respect to the
spacetime metric, so we may write it in the form
\begin{equation}\label{eq:BIa}
a_b^{(11)} = \zeta \big( \delta _b^1 \, \mathrm{cos} \, \varphi 
+ \delta _b^2 \, \mathrm{sin} \, \varphi \big) + \gamma \, p_b
\end{equation}
with scalar coefficients $\zeta$ and $\gamma$ and an angle $\varphi$ which gives,
at each point along the ray, the rotation of the polarization plane with respect
to the $(x^1,x^2)$ basis vectors which are parallel with respect to the 
Levi-Civita derivative of the metric along the ray.

For the electromagnetic field (\ref{eq:BIF}), the partial
derivatives of the Lagrangian are
\begin{equation}\label{eq:BILF}
\mathcal{L} _F = \dfrac{- \, 1 }{2 \sqrt{1+ \dfrac{B_0^2}{b_0^2}}
\sqrt{ 1 - \dfrac{E (t) ^2}{c^2b_0^2} }}
\, ,
\end{equation}
\begin{equation}\label{eq:BILG}
\mathcal{L} _G = \dfrac{- \, B_0 E ( t )}{
c \, b_0^2 \sqrt{1+ \dfrac{B_0^2}{b_0^2}}
\sqrt{ 1 - \dfrac{E (t) ^2}{c^2b_0^2} }}
\, .
\end{equation}
With these results, inserting (\ref{eq:BIa}) into (\ref{eq:transprinc}) yields
\begin{equation}\label{eq:dotphi}
\dot{\varphi} = 
\dfrac{- \, B_0 \dfrac{d E(t)}{dt}}{2 \, c \, b_0^2 
\Big( 1 - \dfrac{E (t)^2}{c^2b_0^2}\Big)}
\end{equation}
where the overdot means derivative with respect to $t$ along the ray, 
$p^a \nabla _a \varphi = \dot{\varphi}$. We see that the time-dependence
of $\mathcal{L}_G$ produces a rotation of the polarization
plane. Integrating (\ref{eq:dotphi}) from $t_1= 0$ with $E(t_1)=0$ to 
a time $t_2$ with $E(t_2)=E_0$ gives 
\begin{equation}\label{eq:phi}
\Delta \varphi = 
\dfrac{B_0}{2 b_0} \, \mathrm{arctanh} \Big( \dfrac{E_0}{c b_0} \Big) 
=
\dfrac{B_0 E_0}{2 c b_0 ^2} \, 
\left( 1 + O \Big( \dfrac{E_0^2}{c^2 b_0^2} \Big) \right) 
\, . 
\end{equation}
In principle, this can be utilized for a new laboratory test of the Born-Infeld 
theory. If in the considered constellation with strong fields $E_0$ and $B_0$ 
no rotation of the polarization plane is detected, this gives a lower
bound on $b_0$. 
However, the effect on the polarization plane is so small that with present-day
technology this test of the Born-Infeld theory is not yet competitive with 
other tests. E.g., in Ref. \cite{SchellstedePerlickLaemmerzahl2015} we have 
seen that with interferometric methods one could find a bound on $b_0$ in the
order of $b_0 \gtrsim 7 \times 10^7$ T which corresponds to $b_0 \gtrsim 7 \times 10^{11} 
\sqrt{\mathrm{g}} / \big( \sqrt{\mathrm{cm}} \, \mathrm{s} \big)$ in Gaussian units.
Assuming that a rotation of the polarization plane by one 
arcminute could be measured, $\Delta \varphi \approx 3 \times 10^{-4}$ rad, we would 
need electric and magnetic fields of 
\begin{equation}\label{eq:E0B0}
\left| B_0 \, \dfrac{E_0}{c} \right| \approx 3 \times 10^{12} \, \mathrm{T}^2
\end{equation}
for being competitive with the interferometric test. This is not reachable in
a laboratory experiment in the near future.

\section{Conclusions}\label{sec:conclusions}

It was the main purpose of this paper to derive the transport law of the
polarization plane in nonlinear vacuum electrodynnamics. We have done
this, on an unspecified general-relativistic spacetime, for a theory of the 
Pleba{\'n}ski class and an electromagnetic background field which were
arbitrary except for some non-degeneracy conditions. To that
end we have utilized an approximate-plane-harmonic-wave ansatz which takes the
generation of higher harmonics and frequency rectification into account.
According to this ansatz, the electromagnetic field is written as an
asymptotic series with respect to a parameter $\alpha$ where the limit
$\alpha \to 0$ refers to sending the frequency to infinity. We have 
seen that the generalized Maxwell equations have to be solved to 
zeroth and to first order with respect to $\alpha$ for determining the 
transport law of the polarization plane in lowest non-trivial order. 

When considering the generalized Maxwell equations to zeroth
order, we have rederived the known result that, for every theory of
the Pleba{\'n}ski class and every background field that satisfy
the assumed non-degeneracy conditions, there are two optical metrics 
which are both of Lorentzian signature. For a solution of the zeroth-order 
equations one needs a scalar function, the so-called eikonal function $S$, 
whose gradient $p_b = \nabla _b S$ is lightlike with respect to at least one of 
the two optical metrics, and a covector field, $a_b^{(11)}$, which has 
to satisfy an algebraic equation known as the zeroth-order polarization 
condition. 

We have seen that two cases have to be distinguished. The first case
is that of a solution of multiplicity one, i.e., the case that $p_b=\nabla _bS$
is lightlike with respect to only one of the optical metrics. Then the 
polarization plane (i.e., the plane spanned by $a_b^{(11)}$ and $p_b$) is
uniquely determined by the zeroth-order polarization condition. The 
first-order equations give no additional information on the polarization
plane and have to be checked only for consistency. The second case is that of 
a solution of multiplicity two, i.e., the case that $p_b = \nabla _b S$ is
lightlike with respect to both optical metrics. Then the 
zeroth-order polarization condition allows an arbitrary rotation of the 
polarization plane along each ray. However, the generalized Maxwell equations
at first order give us a transport equation which determines the polarization 
plane uniquely along a ray if it is given at one point of this ray. This 
transport law has a fairly simple form in the case that $p_b$ is a principal
null covector of the electromagnetic background field, $F_a{}^bp_b \sim p_a$,
see Section \ref{sec:transport1}. It is much more awkward if this is not
the case, see Section \ref{sec:transport2}. 

We have exemplified the general results with the Born-Infeld theory. In this
theory the two optical metrics coincide, i.e., all solutions of the eikonal 
equation are of multiplicity two. We have considered a particular solution 
where $p_b$ is a principal null covector of a background field on Minkowski 
spacetime. In this example, the rays are straight lightlike lines, i.e., the 
light propagation is the same as in the standard Maxwell vacuum theory. However, 
the behavior of the polarization plane is different: Whereas in the standard 
Maxwell theory it is parallely transported along each ray, here it rotates 
with respect to a parallely transported plane by an angle $\Delta \varphi$. 
This is a feature that could be observed in a laboratory experiment, sometimes
in the future, when sufficiently strong fields are available.

\section*{Acknowledgements}
We thank Gerold Schellstede for helpful discussions.
V.P.  is  grateful  to  Deutsche Forschungsgemeinschaft for financial 
support under Grant  No.   LA  905/14-1. Moreover, C.L. and V.P.
gratefully  acknowledge support from the Deutsche Forschungsgemeinschaft 
within the Research Training Group 1620 ``Models of  Gravity.''
A.M. acknowledges support from DAAD.


\end{document}